\newif\ifkuvat
\newif\ifhypercolors

\kuvattrue
\hypercolorstrue

\documentclass[a4paper,12pt]{article}

\usepackage[left=3cm,top=2.5cm,right=3cm,bottom=3cm]{geometry}
\usepackage{setspace}
\usepackage{lmodern}
\usepackage[T1]{fontenc}
\usepackage[utf8]{inputenc}
\usepackage{textcomp}
\usepackage[english]{babel}
\usepackage{enumerate}
\usepackage{tabularx}
\usepackage{changepage}
\usepackage{bigstrut}
\usepackage{multirow}
\usepackage{parskip}

\usepackage{graphicx}
\usepackage{tikz}
\usetikzlibrary{arrows}
\usepackage{caption}
\usepackage{subcaption}
\usepackage[SCI]{aaltologo} 

\usepackage{amsmath}
\usepackage{amssymb}
\usepackage{amsthm}
\usepackage{mathtools}
\usepackage{cases}
\usepackage{enumitem}
\usepackage{pdfpages}
\usepackage{algorithm}
\usepackage[noend]{algpseudocode}

\ifhypercolors
\usepackage[pdftex,colorlinks,pdftitle={Kandityö (Aalto SCI, 2015)},pdfauthor={Antti Pöllänen}]{hyperref}
\hypersetup{linkcolor=black}
\hypersetup{urlcolor=black}
\hypersetup{citecolor=black}
\fi

\newlength{\tavallinenparskip}
\newlength{\tavallinenparindent}

\providecommand{\insertlogosmall}[2]{\ifkuvat\ifhypercolors\AaltoLogoSmall{1}{#1}{#2}\else\AaltoLogoSmall{1}{#1}{aaltoBlack}\fi\fi}
\providecommand{\insertlogolarge}[2]{\ifkuvat\ifhypercolors\AaltoLogoLarge{1}{#1}{#2}\else\AaltoLogoLarge{1}{#1}{aaltoBlack}\fi\fi}

\mathtoolsset{showonlyrefs,showmanualtags} 

\theoremstyle{definition}
\newtheorem{mydef}{Definition}
\theoremstyle{theorem}
\newtheorem{thm}{Theorem}

\newtheorem{lemma}{Lemma}

\newcommand{\cl}{\textrm{cl}}
\newcommand{\te}[1]{\textrm{#1}}
\newcommand{\q}{\textquotesingle}
\newcommand{\kr}{\left\lceil \frac{k}{r} \right\rceil}
\newcommand{\lkr}{\lceil \frac{k}{r} \rceil}
\newcommand{\I}{\mathcal{I}}

\begin{document}

\pagenumbering{roman}

\thispagestyle{empty}

\begin{adjustwidth}{1cm}{1cm}

\vspace*{\fill}

\noindent\insertlogolarge{!}{aaltoYellow}
\vspace*{2cm}

\noindent{\Large Antti Pöllänen}
\vspace*{1cm}

\noindent{\Large Locally Repairable Codes and Matroid Theory}
\vspace*{2cm}

\noindent Bachelor's thesis
\vspace*{.5cm}

\noindent Espoo November 22, 2015
\vspace*{1cm}

Supervisor: Associate Professor Camilla Hollanti

Instructor:  Ph.D. Thomas Westerbäck

\vspace*{2cm}

\end{adjustwidth}

\clearpage

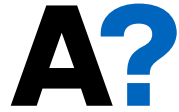
\hfill\insertlogosmall{?}{aaltoRed}

\setlength{\tavallinenparskip}{\parskip}
\setlength{\tavallinenparindent}{\parindent}

\noindent{\renewcommand{\arraystretch}{1.3}\begin{tabularx}{\textwidth}{|X|X|X|X|X|X|}
\hline
\multicolumn{3}{|p{\textwidth/2 - 2\tabcolsep - 3\fboxrule/2}|}{%
Aalto University \newline
School of Science \newline
PL 11000, 00076 Aalto \newline
\ifhypercolors\href{http://sci.aalto.fi/en/}{http://sci.aalto.fi/en/}\else http://www.aalto.fi \fi} &
\multicolumn{3}{p{\textwidth/2 - 2\tabcolsep - 3\fboxrule/2}|}{%
\textsc{Abstract of the bachelor's \newline thesis}} \\
\hline
\multicolumn{6}{|p{\textwidth - 2\tabcolsep - 2\fboxrule}|}{%
\textbf{Author: }Antti Pöllänen} \\
\hline
\multicolumn{6}{|p{\textwidth - 2\tabcolsep - 2\fboxrule}|}{%
\textbf{Title: }\newline
Locally Repairable Codes and Matroid Theory} \\
\hline
\multicolumn{6}{|p{\textwidth - 2\tabcolsep - 2\fboxrule}|}{%
\textbf{Degree programme: }Engineering Physics and Mathematics} \\
\hline
\multicolumn{3}{|p{\textwidth/2 - 2\tabcolsep - 3\fboxrule/2}|}{%
\textbf{Major subject: }Mathematics and Systems Analysis} &
\multicolumn{3}{p{\textwidth/2 - 2\tabcolsep - 3\fboxrule/2}|}{%
\textbf{Major subject code: }SCI3029    } \\
\hline
\multicolumn{6}{|p{\textwidth - 2\tabcolsep - 2\fboxrule}|}{%
\textbf{Supervisor: }Associate Professor Camilla Hollanti \newline
\textbf{Instructor: }Ph.D. Thomas Westerbäck} \\
\hline
\multicolumn{6}{|p{\textwidth - 2\tabcolsep - 2\fboxrule}|}{%
\textbf{Abstract: }\newline
Locally repairable codes (LRCs) are error correcting codes used in distributed data storage. A traditional approach is to look for codes which simultaneously maximize error tolerance and minimize storage space consumption. However, this tends to yield codes for which error correction requires an unrealistic amount of communication between storage nodes. LRCs solve this problem by allowing errors to be corrected locally.

This thesis reviews previous results on the subject presented in \cite{artikkeli}. These include that every almost affine LRC induces a matroid such that the essential properties of the code are determined by the matroid. Also, the generalized Singleton bound for LRCs can be extended to matroids as well. Then, matroid theory can be used to find classes of matroids that either achieve the bound, meaning they are optimal in a certain sense, or at least come close to the bound. This thesis presents an improvement to the results of \cite{artikkeli} in both of these cases.} \\
\hline
\multicolumn{2}{|p{\textwidth/3 - 2\tabcolsep - 4\fboxrule/3}|}{%
\textbf{Date: }\newline November 22, 2015} &
\multicolumn{2}{p{\textwidth/3 - 2\tabcolsep - 4\fboxrule/3}|}{%
\textbf{Language: }\newline English} &
\multicolumn{2}{p{\textwidth/3 - 2\tabcolsep - 4\fboxrule/3}|}{%
\textbf{Number of pages: } \newline $4+28$} \\
\hline
\multicolumn{6}{|p{\textwidth - 2\tabcolsep - 2\fboxrule}|}{%
} \\
\hline
\end{tabularx}}

\clearpage
\hfill\insertlogosmall{?}{aaltoBlue}

\setlength{\tavallinenparskip}{\parskip}
\setlength{\tavallinenparindent}{\parindent}

\noindent{\renewcommand{\arraystretch}{1.3}\begin{tabularx}{\textwidth}{|X|X|X|X|X|X|}
        \hline
        \multicolumn{3}{|p{\textwidth/2 - 2\tabcolsep - 3\fboxrule/2}|}{%
            Aalto-yliopisto \newline
            Perustieteiden korkeakoulu \newline
            PL 11000, 00076 Aalto \newline
            \ifhypercolors\href{http://www.aalto.fi}{http://sci.aalto.fi/fi/}\else http://sci.aalto.fi/fi/ \fi} &
        \multicolumn{3}{p{\textwidth/2 - 2\tabcolsep - 3\fboxrule/2}|}{%
            \textsc{Kandidaatintyön tiivistelmä}} \\
        \hline
        \multicolumn{6}{|p{\textwidth - 2\tabcolsep - 2\fboxrule}|}{%
            \textbf{Tekijä: }Antti Pöllänen} \\
        \hline
        \multicolumn{6}{|p{\textwidth - 2\tabcolsep - 2\fboxrule}|}{%
            \textbf{Työn nimi: }\newline
            Paikallisesti korjaavat koodit ja matroiditeoria} \\
        \hline
        \multicolumn{6}{|p{\textwidth - 2\tabcolsep - 2\fboxrule}|}{%
            \textbf{Koulutusohjelma: }Teknillinen fysiikka ja matematiikka} \\
        \hline
        \multicolumn{3}{|p{\textwidth/2 - 2\tabcolsep - 3\fboxrule/2}|}{%
            \textbf{Pääaine: }Matematiikka ja systeemitieteet} &
        \multicolumn{3}{p{\textwidth/2 - 2\tabcolsep - 3\fboxrule/2}|}{%
            \textbf{Pääaineen koodi: }SCI3029} \\
        \hline
        \multicolumn{6}{|p{\textwidth - 2\tabcolsep - 2\fboxrule}|}{%
            \textbf{Vastuuopettaja: }Professori Camilla Hollanti \newline
            \textbf{Ohjaaja: }Ph.D. Thomas Westerbäck} \\
        \hline
        \multicolumn{6}{|p{\textwidth - 2\tabcolsep - 2\fboxrule}|}{%
            \textbf{Tiivistelmä: }\newline
            Paikallisesti korjaavat koodit ovat virheenkorjauskoodeja, joita käytetään hajautetuissa tallennusjärjestelmissä. Perinteisesti on etsitty koodeja, jotka mahdollistavat mahdollisimman monen yhtäaikaisen virheen korjaamisen ja samanaikaisesti kasvattavat tallennustilan tarvetta mahdollisimman vähän.
            Tällaisilla koodeilla virheenkorjaus edellyttää kuitenkin usein epärealistisen paljon kommunikaatiota tallennusyksiköiden välillä. Paikallisesti korjaavien koodien tarkoitus on ratkaista tämä ongelma tekemällä virheiden korjaamisesta paikallista.
            
            Tässä työssä selostetaan aiheeseen liittyvät artikkelissa \cite{artikkeli} esitetyt tutkimustulokset. Niihin lukeutuu, että kutakin lähes affiinia paikallisesti korjaavaa koodia vastaa yksikäsitteisesti matroidi, josta ilmenevät koodin olennaiset ominaisuudet. Lisäksi yleistetty Singleton-raja paikallisesti korjaaville koodeille voidaan yleistää koskemaan myös matroideja. Näiden tulosten avulla matroiditeoriaa voidaan hyödyntää sellaisten matroidiluokkien löytämisessä, jotka joko saavuttavat Singleton-rajan, eli ovat tietyssä mielessä optimaalisia, tai ainakin yltävät lähelle sitä. Työssä löydetään parannus aiempiin tuloksiin kummassakin näistä tapauksista.} \\
        \hline
        \multicolumn{2}{|p{\textwidth/3 - 2\tabcolsep - 4\fboxrule/3}|}{%
            \textbf{Päivämäärä: }22.11.2015} &
        \multicolumn{2}{p{\textwidth/3 - 2\tabcolsep - 4\fboxrule/3}|}{%
            \textbf{Kieli: }englanti} &
        \multicolumn{2}{p{\textwidth/3 - 2\tabcolsep - 4\fboxrule/3}|}{%
            \textbf{Sivumäärä: }$4+28$} \\
        \hline
        \multicolumn{6}{|p{\textwidth - 2\tabcolsep - 2\fboxrule}|}{%
            \textbf{Avainsanat: }hajautetut tallennusjärjestelmät, matroidi, pyyhkiymäkanava, paikallisesti korjaava koodi, lähes affiini koodi, yleistetty Singleton-raja} \\
        \hline
    \end{tabularx}}

\clearpage

\begin{spacing}{1.6}
\tableofcontents
\end{spacing}
\clearpage

\pagenumbering{arabic}

\section{Introduction}

In modern times, the need for large scale data storage is swiftly increasing. This need is present for example in large data centers and in cloud storage. The large scale of these distributed data storage systems makes hardware failures common. However, the data must not be lost, and therefore means to recover corrupted data must be devised.

Coding theory can be used as a tool for solving this problem. Coding refers to the process of converting the data into a longer redundant form in such a way that errors occurred after the coding can be corrected.

There are various different codes that could be used in the context of distributed storage. However, in this paper we are interested in a class of codes called \emph{locally repairable codes} (LRCs). Using these codes we can optimize not only storage space consumption and global error tolerance, but also local error tolerance. Local error tolerance or correction is desirable because it reduces the need for communication between storage units.


Every \emph{almost affine} LRC induces a \emph{matroid} such that the parameters $(n,k,d,r,\delta)$ of the LRC appear as matroid invariants. Consequently, the \emph{generalized Singleton bound} for the parameters $(n,k,d,r,\delta)$ of an LRC can be extended to matroids. Matroid theory can then be utilized to design LRCs that achieve the bound or at least come close to it. We review these results first introduced in \cite{artikkeli} as well as present two improvements to them.


\section{Locally repairable codes}

\subsection{Basics}
By a locally repairable code we mean a \emph{block code} with certain local error correction properties. Let us start by reviewing the basic concept of a block code. The coding procedure using a block code can be defined as an injective mapping $\gamma: M \rightarrow \Sigma^n$, where $M$ is the set of symbols used to represent the non-coded information and $\Sigma$ is the alphabet used to represent the coded information. By the term \emph{code} we mean the set $C=\gamma(M) \subsetneq \Sigma^n$, i.e. the image of $\gamma$. The code is a proper subset of $\Sigma^n$ because redundancy must be introduced by the coding process in order for error detection or correction to be possible. By a \emph{codeword} we mean an element $c \in C$. The length of a codeword, i.e. the \emph{block length} of the code, is denoted by $n$. The size of the alphabet $\Sigma$ is denoted by $q$ ($q=|\Sigma|$).

The \emph{dimension} of the code $k$ is defined by $k=\log_q(|C|)$. This means that if the original information is presented using the same alphabet $\Sigma$ as the coded information, $k$ symbols are needed to present the non-coded information and we have $M = \Sigma^k$.

The \emph{rate} $R$ of a code is given by $R=\frac{k}{n} < 1$. A high rate is desirable because then less code symbols are needed to convey the coded information.

The \emph{Hamming distance} $\Delta(\mathbf{x}, \mathbf{y})$ of two codewords $\mathbf{x}, \mathbf{y} \in C$ is defined as the number of positions at which the two codewords differ. The \emph{minimum distance} $d$ of a code $C$ is defined by $d = \min_{\mathbf{x}, \mathbf{y} \in C, \mathbf{x} \neq \mathbf{y}}(\Delta(\mathbf{x}, \mathbf{y}))$. A large minimum distance is desirable because then more simultaneous errors can be corrected. 

In this paper we use an\textit{ erasure channel model}, which means that the potential errors include only \textit{erasures}: Each element of the codeword either stays correct or is erased, and we know what the indices of the erased elements are. In  this case, $d-1$ simultaneous errors can always be corrected. 


Finally, we need the concept of a \emph{projection of a code}. For any subset $X=\{i_1,...,i_m\} \subseteq [n]$, the projection of the code $C$ to $\Sigma^{|X|}$, denoted by $C_X$, is the set

\begin{equation*}
C_X = \{(c_{i_1},...,c_{i_{m}}): \mathbf{c} = (c_1,...,c_n) \in C\}.
\end{equation*}

The expression $[n]$ for $n \in \mathbb{Z}_+$ stands for the set of integers from 1 to $n$, i.e. \newline $[n] = \{y \in \mathbb{Z}_+: y \leq n \}$.

\subsection{Locally repairable codes}

A large rate and  minimum distance are not always the only design criteria for a good code. With locally repairable codes we are also interested in a certain kind of locality of the error correction which is described by the parameters $r$ and $\delta$. Let us next give the definition of an \emph{$(n,k,d,r,\delta)$-LRC} (originally defined in \cite{LRCpaper}).

First we need the notion of an \emph{$(r,\delta)$-locality set}: \newline

\begin{mydef}
    When $1\leq r \leq k$ and $\delta \geq 2$, an $(r,\delta)$-locality set of $C$ is a subset $S\subseteq [n]$ such that
    
    \begin{equation}
    \label{locSet}
    \begin{alignedat}{3}
    & \textrm{(i)} \quad && |S|\leq r+\delta-1, \\
    & \textrm{(ii)} \quad && l \in S, L=\{i_1,...,i_{|L|}\} \subseteq S \setminus \{l\} \textrm{ and } |L|=|S|-(\delta-1) \Rightarrow \\ &\quad \quad && \exists f: C_L\rightarrow  \Sigma \textrm{ such that } f((c_{i_1},...,c_{i_{|L|}})) = c_l \te{ for all } \mathbf{c} \in C.
    \end{alignedat}
    \end{equation}
\end{mydef}

This definition implies that any $\delta-1$ code symbols of a locality set can be recovered by the rest of the symbols of the locality set and the locality set is at most of size $r+\delta-1$. This also means that any $|S|-(\delta-1)$ code symbols of the locality set can be used to determine the rest of the symbols in the locality set.

The condition (ii) could also be equivalently expressed as either of the following:

\begin{equation*}
\begin{alignedat}{3}
& \te{\emph{(ii)\q}} \quad && l \in S, L=\{ i_1,...,i_{|L|}\} \subseteq S \setminus \{l\} \te{ and } |L|=|S|-(\delta-1) \Rightarrow \\
& && |C_{L \cup \{l\}}| = |C_L|,\\
& \te{\emph{(ii)\q \q}} \quad && d(C_S) \geq \delta \te{, where $d(C_S)$ is the minimum distance of $C_S$}.
\end{alignedat}
\end{equation*}

We say that $C$ is a \emph{locally repairable code} with \emph{all-symbol locality} $(r,\delta)$ if every code symbol $l \in [n]$ is included in an $(r,\delta)$-locality set.

\subsection{The Singleton bound}
Let us from now on continue our analysis in the context of linear codes. An \emph{(n,k)-linear code} $C$ is a linear subspace of $\mathbb{F}_q^n$ with dimension $k$. Here $\mathbb{F}_q$ is a finite field of size $q$ (which is a prime) and acts as the alphabet of the code. If we assume that the original information is presented using the same alphabet $\mathbb{F}_q$, we get a linear code via the following encoding function $\gamma: \mathbb{F}_q^k \rightarrow \mathbb{F}_q^n$:

\begin{equation*}
\gamma (\mathbf{x}) = \mathbf{x}^T \mathbf{G}, \quad \textrm{ where } \mathbf{x} \in \mathbb{F}_q^k \textrm{ and } \mathbf{G} \in \mathbb{F}_q^{k \times n}.
\end{equation*}

Here $\mathbf{G}$ is called a \emph{generator matrix} of the linear code.  

So far we have stated the following criteria that are desirable for any code: a large rate $R=\frac{k}{n}$ and a large minimum distance $d$. These objectives are clearly contradictory, since a good rate implies only little redundancy in the code, whereas having a large minimum distance forces the code to have a lot of redundancy. For linear codes, this tradeoff is described by the \emph{Singleton bound}:
\begin{equation}
\label{singleton}
d \leq n-k+1.
\end{equation}

For linear locally repairable codes, we obtain the \emph{generalized Singleton bound} \cite{LRCpaper}:

\begin{equation}
\label{genSingleton}
d \leq n-k+1- \left( \left\lceil\frac{k}{r} \right\rceil -1\right)(\delta-1).
\end{equation}

The notation $\lceil \cdot \rceil$ denotes rounding up to the nearest integer. Similarly we will use $\lfloor \cdot \rfloor$ to denote rounding down to the nearest integer.  Note that this is a less strict bound than \eqref{singleton}, since every linear code achieving the bound \eqref{singleton} (i.e. satisfying it as an equation) also achieves the bound \eqref{genSingleton}, but a code achieving \eqref{genSingleton} only achieves \eqref{singleton} if $k=r$ (as we assume $\delta \geq 2$).

From now on, we will refer to the generalized Singleton bound \eqref{genSingleton} merely as the \emph{Singleton bound} or even as the \emph{bound} when the meaning is clear from context. When a code satisfies \eqref{genSingleton} as an equation, we say that the code achieves the bound or that the code is optimal. By optimal we do not here mean optimality in any objective sense. Instead, we could use \emph{Pareto optimality} as an objective criterion for the  favorableness of an LRC. By a Pareto optimal LRC we mean an LRC for which it is impossible to improve one parameter without weakening some other, with a larger $d$, $k$ and $\delta$ as well as a smaller $n$ and $r$ always being more desirable. The Pareto optimal LRCs are a subset of the LRCs with maximal $d$, i.e. the $(n,k,d,r,\delta)$-LRCs for which there exists no LRC with parameters  $(n,k,d',r,\delta)$ such that $d' > d$. The significance of LRCs achieving the bound is that they are Pareto optimal except for some cases where decreasing $r$ would not alter the value of the right side of inequality \eqref{genSingleton} or where $\lkr-1=0$ enabling a suboptimal $\delta$.

\section{Matroids}

Matroids are abstract combinatorial structures that capture a certain mathematical notion of dependence that is common to a surprisingly large number of mathematical entities. For example, a set of vectors along with the concept of linear independence yields a matroid. One possibility to define a matroid is the following definition via independent sets \cite{oxley}: \newline

\begin{mydef}
    
    A matroid $M=(E,\mathcal{I})$ is a finite set $E$ with a collection of subsets $\I \subseteq \mathcal{P}(E)$ such that

    \begin{equation}
    \label{matAxioms1}
    \begin{alignedat}{3}
    & \textrm{(i)} \quad && \emptyset  \in \I, \\
    & \textrm{(ii)} \quad && Y \in \I \textrm{ and } X \subseteq Y \Rightarrow X \in \I, \\
    & \textrm{(iii)} \quad && X, Y \in \I \textrm{ and } |X| > |Y| \Rightarrow  \exists x \in X\setminus Y: \{x\}\cup Y \in \I.
    \end{alignedat}
    \end{equation}
\end{mydef}

Here $\mathcal{P}(E)$ denotes the power set of $E$, i.e. $\mathcal{P}(E) = \{Y: Y \subseteq E \}$. We say that a set $X \subseteq E$ is \emph{independent} if $X \in \I$, otherwise it is \emph{dependent}. It is easily verifiable that a set of vectors $E$ along with its linearly  independent subsets $\I$ satisfies this definition. We call a matroid that arises from the column vectors of a matrix a \emph{matric matroid}. Another common class of matroids are those arising from undirected graphs, in which case $E$ is the set of all edges of the graph, and a subset of edges is independent if it does not contain a cycle. The properties in \eqref{matAxioms1} are satisfied for such \emph{graphic matroids} as well. For the definitions of graph theoretic concepts and additional information on graphs, we refer the reader to \cite{graphs}.

There are various mathematical concepts associated with matroids. These concepts are often analogous to concepts already familiar in the context of structures giving rise to matroids, e.g. the column vectors of a matrix or an undirected graph.

Let us start by defining the \emph{rank function} of a matroid $M$: \newline

\begin{mydef}
    \label{rankViaIndSets}
    The rank function $\rho$ of a matroid $M=(E,\I)$ is a function \newline  $\rho: \mathcal{P}(E) \rightarrow \mathbb{Z}$ satisfying the following for $X \subseteq E$:
    \begin{equation*}
    \rho(X) = \max\{|Y|: Y \subseteq X \textrm{ and } Y \in \I\}.
    \end{equation*}
    
\end{mydef}

For matrices, this concept is analogous to the rank of the matrix formed by the column vectors of $X$. For undirected graphs, rank tells the amount of edges in a minimal spanning forest of the subgraph induced by $X$.

The rank function satisfies the following properties \cite{oxley}:\newline

\begin{thm}
    \label{matViaRank}
    
    Let $\rho$ be the rank function of a matroid $M=(E,\I)$. Then for $X,Y \subseteq E$:
    
    \begin{equation}
    \label{rankEqs}
    \begin{alignedat}{3}
    & \textrm{(i)} \quad && 0 \leq \rho (X) \leq |X|, \\
    & \textrm{(ii)} \quad && X \subseteq Y \Rightarrow \rho (X) \leq \rho (Y), \\
    & \textrm{(iii)} \quad && X,Y \subseteq E \Rightarrow \rho(X) + \rho(Y) \geq  \rho (X \cup Y) + \rho(X \cap Y), \\
    & \textrm{(iv)} \quad && X \in \I \Leftrightarrow \rho (X) = |X|.
    \end{alignedat}
    \end{equation}
  
\end{thm}

Equation \eqref{rankEqs}(iii) is called the \emph{semimodular inequality}, and it can for instance be viewed as an upper bound for the rank of a union of two sets. It also implies subadditivity for  the rank function: we have $\rho(\bigcup_{X \in S} X) \leq \sum_{X \in S} \rho(X)$ for any collection of subsets $S \subseteq \mathcal{P}(E)$.
Actually we can use these properties of the rank function as an alternative matroid definition \cite{oxley}: \newline

\begin{mydef}
    A matroid $M=(E,\rho)$ is a finite set $E$ together with a function $\rho: \mathcal{P}(E) \rightarrow \mathbb{Z}$ such that it satisfies the conditions (i)-(iii) in \eqref{rankEqs}.
\end{mydef}
In this case, an independent set is defined by condition (iv), and we get the conditions in Definition \ref{matAxioms1} as a theorem. In this way, these two definitions are equivalent to each other and we can use them interchangeably.

There is also a plethora of further ways to define a matroid. Many of these definitions are constructed using propositions that apply to matroid concepts defined below. This variety of definitions is a truly useful property of matroids: a matroid can be identified using any definition after which we automatically know that the conditions in the other definitions are also satisfied.

Let us next define some of these matroid concepts. The \emph{nullity} of a set $X \subseteq E$ is defined by $\eta(X) = |X|-\rho(X)$. A \emph{circuit} is a dependent set $X \subseteq E$ whose all proper subsets are independent, i.e.  $\rho(X \setminus \{x\}) = \rho(X) = |X|-1 $ for every $x \in X$. This concept is perhaps most easily understood in the context of graphic matroids since then a set of edges is a circuit if and only if it is a cycle of the undirected graph. We denote the set of circuits of a matroid by $\mathcal{C}(M)$.

The \emph{closure} of a set $X \subseteq E$ is defined by $\cl(X) = \{x \in E: \rho(X \cup \{x\}) = \rho(X)\}$. In terms of a matric matroid, the closure $\cl(X)$ consists of all vectors of $E$ that are in the span of the vectors in $X$. For a graphic matroid, the closure is obtained by adding all the edges of $E$ whose endpoints are connected by a walk in the graph or whose two endpoints are the same vertex. 

A set $X\subseteq E$ is $cyclic$ if it is a union of circuits. Alternatively a cyclic set is a set $X$ such that $\forall x \in X: \rho(X \setminus \{x\}) = \rho(X)$. For matric matroids, this means that $X$ includes no vector that would not be in the span of the rest of the vectors. We denote the set of cyclic sets of a matroid by $\mathcal{U}(M)$.

A set $X\subseteq E$ is a \emph{flat} if $X=\cl(X)$. A \emph{cyclic flat} is a flat that also is a cyclic set, i.e. a union of circuits.

The cyclic flats of a matroid have the property that they form a \emph{finite lattice} $\mathcal{Z}$ with the following \emph{meet} and \emph{join} for $X,Y \in \mathcal{Z}$ \cite{latticeOfCyclicFlats}:

\begin{align*}
& X\wedge Y = \bigcup_{C \in \mathcal{C}(M): C \subseteq X\cap Y}, \\ \vspace{0.25 cm}
& X \vee Y = \cl(X \cup Y).
\end{align*}

The set of the atoms of the lattice is denoted by $A_{\mathcal{Z}}$ and the set of the  coatoms by $coA_{\mathcal{Z}}$. We refer the reader unfamiliar with partial orders and order-theoretic lattices to \cite{artikkeli} for a minimal background or to \cite{lattices} for a more comprehensive exposition. Another way to define a matroid is via this \emph{lattice of cyclic flats}. In fact, this viewpoint is very useful for us since we are later using the lattice of cyclic flats as a tool for constructing and analyzing matroids that correspond to good LRCs. The associated axioms are presented in the following theorem, where $0_{\mathcal{Z}}$ denotes the least element and $1_{\mathcal{Z}}$ denotes the greatest element of the finite lattice $\mathcal{Z}$:\newline

\begin{thm} (\cite{latticeOfCyclicFlats})
    Let $\mathcal{Z} \subseteq \mathcal{P}(E)$ and let $\rho$ be a function $\rho: \mathcal{Z}\rightarrow \mathbb{Z}$. There is a matroid $M$ on $E$ for which $\mathcal{Z}$ is the set of cyclic flats and $\rho$ is the rank function restricted to the sets in $\mathcal{Z}$ if and only if
    
    \begin{equation}
    \begin{alignedat}{2}
    (Z0) \quad & \te{$\mathcal{Z}$ is a lattice under inclusion,} \\
    (Z1) \quad & \rho(0_{\mathcal{Z}})= 0, \\
    (Z2) \quad & X,Y \in \mathcal{Z} \te{ and } X \subsetneq Y \Rightarrow \\
    & 0 < \rho(Y)-\rho(X)<|Y|-|X|,\\
    (Z3) \quad & X,Y \in \mathcal{Z} \Rightarrow\rho(X) + \rho(Y) \geq \\
    & \rho(X \vee Y) + \rho(X \wedge Y) + |(X \cap Y) \setminus (X \wedge Y)|.
    \end{alignedat}
    \end{equation}
    
\end{thm}

The \emph{restriction} of $M$ to $X$ is the matroid $M|X = (X, \rho_{|X})$ with $\rho_{|X}(Y)=\rho(Y)$ for all subsets $Y \subseteq X$. This is clearly a matroid since $\rho$ satisfying the conditions \eqref{rankEqs} (i)-(iii) implies $\rho_{X|}$ also satisfying them. Similarly, a subset of $X$ is independent for $M|X$ exactly if it is independent for the original matroid $M$.

For each matroid $M = (\rho, E)$ there is a \emph{dual matroid} $M^* = (\rho^*, E)$ defined by $\rho^*(X) = \rho(E\setminus X) + |X| - \rho(E)$ (which can be proved by checking conditions (i)-(iii) in \eqref{matViaRank}). This means that a subset $X\subseteq E$ is independent for $M^*$ if and only if $\rho(E\setminus X) = \rho(E)$ i.e. $X \subseteq \cl(E\setminus X)$ for $M$. The circuits of $M^*$ are the minimal sets $C^* \subseteq E$ such that $\rho(E \setminus C^*) < \rho(E)$.

Lastly, we need the concept of a uniform matroid. An \emph{$(n,k)$-uniform matroid} is a matroid $M=(E,\I)$ for which $|E|=n$ and a set $X \subseteq E$ is independent if and only if $|X| \leq k$. Defined via the rank function, a matroid is $(n,k)$-uniform if $|E|=n$ and $\rho(X) = \min\{ |X|, k\}$ for all $X \subseteq E$. Note that $\rho(E)=k$. An $(n,k)$-uniform matroid is obtained for instance as the matric matroid of $n$ randomly chosen uniformly distributed vectors in $\{ x \in \mathbb{R}^k:  ||x||< a\}$ with $0<a \in \mathbb{R}$ and where $||\cdot||$ denotes the Euclidean norm. In this case, we get the desired matroid with probability 1.


\section{Almost affine LRCs and their connection to matroids}

The results in this paper apply to a class of codes called \emph{almost affine codes}, with the following definition: \newline

\begin{mydef}
    A code $C \subseteq \Sigma^n$, where $\Sigma$ is a finite set of size $s \geq 2$, is almost affine if for each $X \subseteq [n]$:
    \begin{equation*}
    \log_s(|C_X|) \in \mathbb{Z}.
    \end{equation*}
    
\end{mydef}

The projections $C_X$ of an almost affine code $C$ are also almost affine. The linear codes are a special case of almost affine codes.

The following theorem is the basis for our application of matroid theory to find good LRCs: \newline

\begin{thm} (\cite{affinematroidlink})
    \label{matroidFromAlmostAffine}
    Every almost affine code $C \subseteq \Sigma^n$ with $s = |\Sigma|$ induces a matroid $M_C = ([n],\rho_C)$, where
    \begin{equation}
    \label{eqMatroidFromAlmostAffine}
    \rho_C (X) = \log_s(|C_X|).
    \end{equation}
\end{thm}

However, not every matroid is a matroid induced by an almost affine code. Examples of matroids not obtainable from almost affine codes are presented in \cite{nonAlmAffMatroid}. Later in this paper (Theorem \ref{thmCorollIII.1}) we present a subclass of matroids for which it is shown that there exists a corresponding almost affine code.

Note that according to this theorem, the matroid induced by a linear code is the matric matroid induced by the columns of its generator matrix.

\subsection{Parameters $(n,k,d,r,\delta)$ as matroid invariants}

The remarkable theorem that allows us to analyze the parameters $(n,k,d,r,\delta)$ of an LRC via its associated matroid is the following \cite{artikkeli}: \newline

\begin{thm} (\cite{artikkeli})
    Let $C$ be an almost affine LRC with the associated matroid $M_C=( [n],\rho_C)$. Then the parameters $(n,k,d,r,\delta)$ of $C$ are matroid invariants, where
    \label{thmParamsFromMatroids} \newpage
    \begin{equation}
    \label{paramsFromMatroidsEq}
    \begin{alignedat}{4}
    &\textrm{(i)} && k = \rho_C([n]), \\
    &\textrm{(ii)} && d = \min\{|X|:X \in \mathcal{C}(M^*)\},  \\
    &\textrm{(iii)} \quad && \textrm{$C$ has all-symbol locality $(r,\delta)$ if and only if for every $j\in [n]$  there exists a} \\
    & &&\textrm{subset $S_j \subseteq [n]$ such that } \\
    &  &&\textrm{a) } j \in S_j, \\
    & &&\textrm{b) } |S_j|\leq r+\delta-1, \\
    & && \te{c) } d(C_{S_j}) = \min\{|X|: X \in \mathcal{C}((M_C|S_j)^*)\} \geq \delta, \\
    & && \te{or the equivalent condition to (c) that} \\
    &  &&\te{c') } \textrm{For all $L \subseteq S_j$ with $|L|=|S_j|-(\delta-1)$, and all $l\in S_j \setminus L$,} \\
    &  &&\quad \textrm{we have $\rho_C (L\cup l) = \rho_C(L) $.}
    \end{alignedat}
    \end{equation}
\end{thm}

We will not give a complete proof here but the main ideas of why this result holds. For $n$, $[n]=E$ and $n = |E|$ follow directly from the definitions in Theorem \ref{matroidFromAlmostAffine}. The result $k=\rho_C([n])$ follows straight from \eqref{eqMatroidFromAlmostAffine} since by choosing $X=[n]$, the right side of  the equation is the same as the definition for $k$. In \cite{affinematroidlink} it is proven that $d$ equals the minimum size of a circuit of the dual matroid $M_C^*$
and that $M_{C_x}= M|X$ for $X \subseteq [n]$. Because $C_X$ is also almost affine, it follows that

\begin{equation*}
d(C_X) = \min\{|X|:X \in \mathcal{C}((M|X)^*)\} ,
\end{equation*}

where $d(C_X)$ denotes the minimum distance of $C_X$. An equivalent condition to condition (iii) c) is that for every set $X \subseteq S_j$ for which $|X| \leq \delta-1$, we have $\rho(S_j\setminus X) = \rho(S_j)$, according to our considerations for dual matroids in Section 3. The condition c') is easily seen to be equivalent to this. 

Now we can view the above results as definitions for the parameters  $(n,k,d,r,\delta)$ for matroids. From the viewpoint of the lattice of cyclic flats the parameters are obtained as follows: \newline

\begin{thm} (\cite{artikkeli})
    Let $M=(E,\rho)$ be a matroid with $0 < \rho(E)$ and $1_{\mathcal{Z}} = E$. Then
    \label{thmParamsFromCyclicFlats}
    \begin{equation}
    \label{eqParametersFromCyclicFlats}
    \begin{alignedat}{3}
    &(i) \quad && n=|1_{\mathcal{Z}}|, \\
    &(ii) \quad && k = \rho(1_{\mathcal{Z}}), \\
    &(iii) \quad && d=n-k+1-\max\{\eta(Z): Z \in coA_{\mathcal{Z}}\}, \\
    &(iv) \quad && \te{$M$ has locality $(r,\delta)$ if and only if for each $x \in E$ there exists a cyclic set} \\
    & && \te{$S_x \in \mathcal{U}(M)$ such that} \\
    & && \te{a) } x \in S_x, \\
    & && \te{b) } |S_x| \leq r+\delta-1, \\
    & && \te{c) } d(M|S_x) = \eta(S_x)+1-\max\{\eta(Z): Z \in coA_{\mathcal{Z}(M|S_x)}\} \geq \delta.
    \end{alignedat}
    \end{equation}
\end{thm}

The only non-trivial parts of this theorem are the expressions for $d$ in (iii) and (iv) c. From Theorem \ref{thmParamsFromMatroids} we know that $d$ equals the size of a minimal circuit of the dual matroid, i.e. the size of a minimal set $X \in E$ for which $\rho(E \setminus X) < \rho(E)$. The problem of finding $d$ is thus reduced to finding the maximal $Y = E \setminus X$ such that $Y$ does not have full rank. In \cite{artikkeli} the following result is proved: \newline

\begin{lemma}
    \label{lemmaMaxNullity}
    \te{If $\rho(X) < \rho(E)$ and $1_{\mathcal{Z}} = E$, then $\eta(X) \leq \max\{\eta(Z): Z \in coA_{\mathcal{Z}} \}$}.
\end{lemma}

Let us examine the set $Y'$ that we get by taking a coatom  of maximal nullity $Z_{max}$ and adding elements to it such that it reaches rank $\rho(E)-1$. Every element added to $Z_{max}$ increases rank, due to Lemma \ref{lemmaMaxNullity} (since adding an element always increases either rank or nullity by one). Thus we have $|Y'| = \rho(Y')+ \eta(Y') = \rho(E) - 1 + \eta(Z_{max})$. Using Lemma \ref{lemmaMaxNullity} again, we notice that $Y'$ now has the maximal size of a set with non-full rank. Thus we have $d = |E|-|Y'| = n-k+1-\max\{\eta(Z): Z \in coA_{\mathcal{Z}}\}$.

\subsection{The generalized Singleton bound for matroids}

It turns out that the generalized Singleton bound \eqref{genSingleton} applies to matroids as well \cite{artikkeli}. Let us next present the main ideas of how this bound is obtained. Firstly we have the two following lemmas: \newline

\begin{lemma} (\cite{artikkeli})
    \label{lemmaChain}
    Let $M = (\rho, E)$ be a matroid with parameters $(n,k,d,r,\delta)$ and let $\{S_x\}_{x \in E}$ be a collection of cyclic sets of $M$ for which the conditions (a)-(c) in Theorem \ref{thmParamsFromCyclicFlats} are satisfied. Then there is a subset of cyclic sets $\{S_j\}_{j \in [m]}$ of $\{S_x\}_{x \in E}$  such that for $Y_j = \cl(Y_{j-1} \cup S_j) = Y_{j-1} \vee \cl(S_j)$, where $j=1,...,m$ we have
    
    \begin{equation}
    \begin{alignedat}{3}
    & \te{(i)} \quad && C: 0_{\mathcal{Z}}=Y_0\subsetneq Y_1 \subsetneq ... \subsetneq Y_m = E \te{ is a chain in $(\mathcal{Z}(M), \subseteq)$}, \\
    & \te{(ii)} \quad && \rho(Y_j) - \rho(Y_{j-1})\leq r, \\
    & \te{(iii)} \quad && \eta(Y_j) - \eta(Y_{j-1})\geq \delta-1.
    \end{alignedat}
    \end{equation}

\end{lemma}

A sketch of the proof is the following: There is a subset $\{S_j\}_{j \in [m]}$ satisfying condition (i) as $\bigcup_{x \in E}S_x = E$. The semimodular inequality \eqref{matViaRank} (iii) implies that $\rho(Y_j) \leq \rho(Y_{j-1}) + \rho(S_j)$, since $\rho(Y_j) = \rho(Y_{j-1} \cup \rho(S_j))$. Together with $\rho(S_j)\leq r$ we get condition (ii).  We have $|S_j \setminus Y_{j-1}| \geq \delta$, since otherwise we would have a set $X\subseteq S_j$ with $|X|\leq \delta-1$ and $X \cap Y_{j-1} = \emptyset$ for which $X \subsetneq \cl(S_j \setminus X)$. This is because $Y_{j-1} \cap S_j$ is a subset of a flat. Such a set $X$ would contradict with the definition of a locality set $S_j$.
Now condition (iii) follows because for any $X \subseteq S_j \setminus Y_{j-1}$ with $|X| = \delta-1$, $X \subseteq \cl(Y_j \setminus X)$. \newline

\begin{lemma} (\cite{artikkeli})
    \label{lemmaTwoIneqs}
    Let $M=(E,\rho)$ be a matroid with parameters $(n,k,d,r,\delta)$ and let $C: 0_{\mathcal{Z}}=Y_0\subsetneq Y_1 \subsetneq ... \subsetneq Y_m = E$ be any chain of $(\mathcal{Z}(M), \subseteq)$ given in Lemma \ref{lemmaChain}. Then we have
    
    \begin{equation*}
    d\leq n-k+1-\eta(Y_{m-1}) \quad \te{and} \quad m \geq \kr.
    \end{equation*}
    
\end{lemma}

The second inequality roughly follows from condition (ii) in Lemma  \ref{lemmaChain} which together with the fact that $\rho(S_i) \leq r$ for every locality set $S_i$ implies $k = \rho(E) \leq mr$. The first inequality follows from Theorem \ref{thmParamsFromCyclicFlats} (iii) and Lemma \ref{lemmaMaxNullity}.

Combining Lemma \ref{lemmaChain} and Lemma \ref{lemmaTwoIneqs} we now get the generalized Singleton bound for matroids: \newline
\begin{thm} (\cite{artikkeli})
    \label{thmSingletonForMatroids}
    Let $M=(E,\rho)$ be an $(n,k,d,r,\delta)$-matroid, then
    
    \begin{equation}
    \label{eqSingleton}
    d \leq n-k+1-\left( \kr -1\right)(\delta-1).
    \end{equation}
\end{thm}

\begin{proof}
    By Lemma \ref{lemmaTwoIneqs} we have that $d \leq n-k+1-\eta(Y_{m-1})$. From Lemma \ref{lemmaChain} (iii) and $m \geq \kr$ in Lemma \ref{lemmaTwoIneqs} it now follows that $\eta(Y_{m-1}) \geq (m-1)(\delta-1) \geq \left( \kr -1\right)(\delta-1)$, which yields the desired result.
\end{proof}

When a matroid satisfies \eqref{eqSingleton} as an equation, we say that the matroid achieves the bound or that the matroid is optimal.
\section{A structure theorem}

The rest of this paper will be focused on finding matroids that achieve the generalized Singleton bound or at least come close to it. A good starting point is the following theorem which gives a set of necessary structural properties for a matroid to achieve the bound. We say that a collection of sets $X_1,X_2,...,X_j$ has a nontrivial union if

\begin{equation*}
X_l \nsubseteq \bigcup_{i \in [j]\setminus \{l\}} X_i \te{ for } l \in [j].
\end{equation*}\newline

\begin{thm}
    \label{thmStructure}
    Let $M=(E,\rho)$ be an $(n,k,d,r,\delta) $-matroid with $r<k$ and
    
    \begin{equation*}
    d = n-k+1-\left( \kr -1\right)(\delta-1).
    \end{equation*}
    
    Also, let $\{S_x:x\in E \} \subseteq \mathcal{U}(M)$ be a collection of cyclic sets for which the conditions (a)-(c) in Theorem \ref{thmParamsFromCyclicFlats} are satisfied. Then
    
    \begin{equation*}
    \begin{alignedat}{3}
    & \te{(i) } && 0_\mathcal{Z} = \emptyset, \\
    & \te{(ii) } && \te{for each } x \in E, \\
    & && \te{a) } \eta(S_x)=(\delta-1), \\
    & && \te{b) } S_x \te{ is a cyclic flat and }\mathcal{Z}(M|S_x) = \{X \in \mathcal{Z}(M): X \subseteq S_x \} = \{\emptyset, S_x\}, \\
    & \te{(iii) } && \te{For each collection $F_1,...,F_j$ of cyclic flats in $\{S_x: x \in E\}$ that has a non-} \\
    & && \te{trivial union,} \\
    & && \te{c) } \eta(\bigvee_{i=1}^{j} F_i) = \begin{cases}
    j(\delta-1) & \te{if } j<\lceil\frac{k}{r} \rceil, \\
    n-k \geq \lceil\frac{k}{r} \rceil (\delta-1) & \te{if } j \geq \lceil\frac{k}{r} \rceil,
    \end{cases} \\
    & && \te{d) } \bigvee_{i=1}^{j} F_i = \begin{cases}
    \bigcup_{i=1}^j F_i &\te{if } j<\lceil\frac{k}{r} \rceil, \\
    E & \te{if } j\geq\lceil\frac{k}{r} \rceil,
    \end{cases} \\
    & && \te{e) } \rho(\bigvee_{i=1}^{j} F_i) = \begin{cases}
    |\bigcup_{i=1}^j F_i| - j(\delta-1) & \te{if } j < \lceil\frac{k}{r} \rceil, \\
    k & \te{if } j \geq \lceil\frac{k}{r} \rceil,
    \end{cases} \\
    & && \te{f) } |F_j \cap \bigcup_{i=1}^{j-1} F_i)| \leq |F_j|-\delta \te{ if } j \leq \left\lceil\frac{k}{r} \right\rceil.
    \end{alignedat}
    \end{equation*}
\end{thm}

The following is an outline of the proof: Let us have a chain C: $0_{\mathcal{Z}}=Y_0\subsetneq Y_1 \subsetneq ... \subsetneq Y_m = E$ in $(\mathcal{Z}(M), \subseteq)$ as earlier. In the process of proving Theorem \ref{thmSingletonForMatroids} we obtained that for every such chain we have

\begin{equation}
\label{eqStructureProof}
\begin{alignedat}{2}
 d & \leq n-k+1-\eta(Y_{m-1}) \\
& \leq n-k+1- (m-1)(\delta-1) \\
& \leq n-k+1-\left( \kr -1\right)(\delta-1) .
\end{alignedat}
\end{equation}


Thus in order to achieve the bound we must have $\eta(Y_{m-1}) = (m-1)(\delta-1)$ for every chain $C$, which together with Lemma \ref{lemmaChain} (iii) implies $\eta(\bigvee_{i=0}^{j} Y_i) = j(\delta-1)$ for $j < m$. This in turn implies $0_\mathcal{Z} = \emptyset$, $\eta(S_x) = (\delta-1)$ and $S_x$ being a cyclic flat for every $x\in E$ as well as condition (iii) c) together with $m=\lkr$ which follows from equation \eqref{eqStructureProof}.


The result $\mathcal{Z}(M|S_x) = \{\emptyset, S_x\}$ is required for $\eta(S_x)=(\delta-1)$ to be possible, since otherwise not every $X \subseteq S_x$ with $|X|=\rho(S_x)$ would have $\rho(X)=\rho(S_x)$.
Note that this also implies that $M|S_x$ is a uniform matroid and that $S_x$ is an atom of the lattice of cyclic flats.

Conditions d) and e) are a direct consequence of c) and  that always $m=\lkr$. If (iii) f) was not true for a locality set $F_j$, we would have $\cl(\bigcup_{i=1}^{j-1} F_i) = \bigcup_{i=1}^{j} F_i$ which would contradict with (iii) c).


\section{Matroid constructions}

In this chapter we will give some explicit matroid constructions first introduced in \cite{artikkeli}.  We later use these constructions to prove existence results for matroids with a large $d$.

Construction 1 gives a class of matroids that is beneficial in the sense that the cyclic flats of its matroids have high rank and minimal size, which implies that the coatoms have small nullity. This in turn means that the matroids from this construction have a large $d$. They also have a simple structure which makes analyzing them easier. 

\emph{Construction 1:} Let $F_1,...,F_m$ be a collection of subsets of a finite set $E$ and let us denote $F_I = \bigcup_{i\in I} F_i$ for $I \subseteq [m]$ . Let $k$ be a positive integer and let $\rho$ be a function $\rho: \{F_i\}_{i \in [m]} \rightarrow \mathbb{Z}$ such that

\begin{equation}
\label{eqConstr1}
\begin{alignedat}{3}
&\te{(i)} \quad && \te{$\{F_i\}_{i \in [m]}$ has a nontrivial union, with $F_{[m]}=E$,} \\
&\te{(ii)} \quad && 0< \rho(F_i) < |F_i| \te{ for every } i \in [m], \\
&\te{(iii)} \quad && \te{There exists $I \subseteq [m]$ such that $F_I-\sum_{i \in I} \eta(F_i) \geq k$,} \\
&\te{(iv)} \quad && \te{If $F_I \in \mathcal{Z}_{<k}$ and $j \in [m]\setminus I $, then $|F_I \cap F_j| < \rho(F_j)$,} \\
&\te{(v)} \quad && \te{If $F_I, F_J \in \mathcal{Z}_{<k}$ and $F_{I \cup J} \notin \mathcal{Z}_{<k}$, then $|F_{I \cup J}|- \sum_{t \in I \cup J} \eta(F_t) \geq k$,}
\end{alignedat}
\end{equation}

where we define
\begin{equation*}
\eta(F_i)=|F_i|-\rho(F_i) \te{ for } i \in [m]
\end{equation*}
and
\begin{equation*}
\mathcal{Z}_{<k} = \{F_J: J \subseteq [m] \te{ and } |F_I|-\sum_{i \in I} \eta (F_i) < k \te{ for all } I \subseteq J \}.
\end{equation*}
Now, to a collection of subsets $F_1,...,F_m$ of $E$, integer $k$ and function $\rho$ that satisfy the conditions (i)-(v), we extend the function $\rho$ to a function on $\mathcal{Z}$, by $\mathcal{Z} = \mathcal{Z}_{<k} \cup E$ and
\begin{equation}
\label{eqExtendRho}
\rho(X) = \begin{cases}
|F_I|-\sum_{i \in I} \eta (F_i) &\te{ if } X=F_I \in \mathcal{Z}_{<k}, \\
k & \te{ if } X=E.
\end{cases}
\end{equation} \newline

\begin{thm} (\cite{artikkeli})
    \label{thmConstr1Props}
    Let $F_1, ..., F_m$ be a collection of subsets of a finite set $E$, $k$ a positive integer and $\rho: \{F_i\}_{i \in [m]} \rightarrow \mathbb{Z}$ a function such that the conditions (i)-(v) of Construction 1 are satisfied. Then $\mathcal{Z}$ and $\rho: \mathcal{Z} \rightarrow \mathbb{Z}$, defined in \eqref{eqExtendRho}, define an $(n,k,d,r,\delta)$-matroid $M(F_1, ..., F_m; k; \rho)$ on $E$ for which $\mathcal{Z}$ is the collection of cyclic flats, $\rho$ is the rank function restricted to the cyclic flats and
    
    \begin{equation*}
    \begin{alignedat}{3}
    & (i) \quad && n = |E|, \\
    & (ii) \quad && k = \rho(E), \\
    & (iii) \quad && d = n-k+1 -\max\{\sum_{i \in I}\eta(F_i): F_I \in \mathcal{Z}_{<k} \}, \\
    & (iv) \quad && \delta-1 = \min_{i \in [m]}\{\eta(F_i)\}, \\
    & (v) \quad && r = \max_{i \in [m]}\{\rho(F_i)\}.
    \end{alignedat}
    \end{equation*}
    
    For each $i \in [m]$, any subset $S \subseteq F_i$ with $|S|=\rho(F_i)+ \delta-1$ is a locality set of the matroid. \newline
\end{thm}

Theorem \ref{thmCorollIII.1} gives a subclass of matroids obtainable from Construction 1 for which it is proven in \cite{artikkeli} that its matroids correspond to almost affine LRCs. This result, given in Theorem \ref{thmLRCsFromMatroids}, is required in order to prove existence results on almost affine LRCs using matroids. The only additional requirement in Theorem \ref{thmCorollIII.1} compared to Construction 1 is that the manner in which the atoms $F_i$ can intersect with each other is more restricted, determined by condition (iv).\newline
\begin{thm} (\cite{artikkeli})
    \label{thmCorollIII.1}
    Let $F_1,...,F_m$ be a collection of subsets of a finite set $E$, $k$ a positive integer and $\rho: \{F_i\}_{i \in [m]} \rightarrow \mathbb{Z}$ a function such that
    \begin{equation}
    \label{eqCorollIII.1}
    \begin{alignedat}{3}
    & (i) \quad && 0<\rho(F_i)< |F_i| \te{ for } i \in [m], \\
    & (ii) \quad && F_{[m]} = E, \\
    & (iii) \quad && k \leq F_{[m]}-\sum_{i \in [m]} \eta(F_i), \\
    & (iv) \quad &&  |F_{[m]\setminus \{j\}} \cap F_j|<\rho(F_j) \te{ for all } j \in [m].
    \end{alignedat}
    \end{equation}
    
    Then $F_1, ..., F_m$, $k$ and $\rho$ define a matroid $M(F_1,...,F_m;k;\rho)$ given in Theorem \ref{thmConstr1Props}. \newline
\end{thm}

\begin{thm}
    \label{thmLRCsFromMatroids} (\cite{artikkeli})
    Let $M(F_1,...,F_m;k;\rho)$ be an $(n,k,d,r,\delta)$-matroid that we get in Theorem \ref{thmCorollIII.1}. Then for every large enough field there is a linear LRC over the field with parameters $(n,k,d,r,\delta)$.
\end{thm}

Theorem \ref{thmLRCsFromMatroids} is proved in \cite{artikkeli} by first constructing a directed graph supporting a gammoid isomorphic to the matroid. (See for instance \cite{oxley} for an explanation on gammoids.) The proof is completed by a result  stating that every finite gammoid is representable over every large enough finite field, which is proved in \cite{gammoid}.

The following graph construction was introduced in \cite{artikkeli}. The matroids it yields are a subclass of those obtainable from Theorem \ref{thmCorollIII.1}. The purpose of it is to give a tool for designing matroids with a large $d$. The main idea of it is to restrict the manner the atoms $F_i$ can share elements in such a way that the matroid can be unambiguously described by a weighted undirected graph, together with information on the rank and nullity of each atom. The vertices correspond to atoms and the weights of the edges tell how many elements are shared by the corresponding atoms.

\emph{Graph construction 1} Let $G = G(\alpha, \beta, \gamma; k, r, \delta)$ be a graph with vertices $[m]$ and edges $W$, where $(\alpha, \beta)$ are two functions $[m] \rightarrow \mathbb{Z}$, $\gamma: W \rightarrow \mathbb{Z}$ and $(k,r,\delta)$ are three integers with $0<r<k$ and $\delta \geq 2$, such that

\begin{align}
\label{c1}
\begin{split}
& \textrm{(i)} \quad \textrm{$G$ is a graph with no 3-cycles,} \\
& \textrm{(ii)} \quad  0 \leq \alpha (i) \leq r-1 \textrm{ for } i \in [m], \\
& \textrm{(iii)} \quad  \beta(i) \geq 0 \textrm{ for } i \in [m], \\
& \textrm{(iv)} \quad  \gamma(w) \geq 1 \textrm{ for } w \in W, \\
& \textrm{(v)} \quad  k \leq rm - \sum_{i \in [m]} \alpha(i) - \sum_{w \in W} \gamma(w), \\
& \textrm{(vi)} \quad r - \alpha(i) > \sum_{w=\{i,j\} \in W} \gamma(W) \textrm{ for } i \in [m].
\end{split}
\end{align}

\begin{thm} (\cite{artikkeli})
    \label{c1thm}
    Let $G(\alpha, \beta, \gamma; k, r, \delta)$ be a graph on $[m]$ such that the conditions (i)-(vi) given in \eqref{c1} are satisfied. Then there is an $(n,k,r,d,\delta)$-matroid $M(F_1,...,F_m;k;\rho)$ given in Theorem \ref{thmCorollIII.1} with
    
    \begin{align}
    \begin{split}
    & \textrm{(i)} \quad n = (r+\delta-1)m - \sum_{i \in [m]} \alpha(i) + \sum_{i \in [m]} \beta(i) - \sum_{w \in W} \gamma(w), \\
    & \textrm{(ii)} \quad d = n-k+1 - \max_{I \in V_{<k}} \{(\delta-1)|I|+\sum_{i \in I} \beta(i) \},
    \end{split}
    \end{align}
    
    where
    \begin{equation*}
    V_{<k} = \{I \subseteq [m]: r|I|- \sum_{i \in I} \alpha(i) - \sum_{i,j \in I, w = \{i,j\}\in W} \gamma(w) < k\}.
    \end{equation*}
    
\end{thm}

\section{The maximal value of $d$ for $(n,k,r,\delta)$-matroids}

Let us denote the largest possible $d$ for a matroid with the parameters $(n,k,r, \delta)$ by $d_{max} = d_{max}(n,k,r,\delta)$. In this chapter, we will review the results in \cite{artikkeli} on  $d_{max}$ and on matroid constructions that yield matroids with large $d$. We will also present two improvements to these results.

The complete function $d_{max}$ is unknown, but in \cite{artikkeli} two kinds of results on it are presented: Firstly, it is proved that for some classes of parameters $(n,k,r,\delta)$ the generalized Singleton bound \eqref{eqSingleton} can be reached and for some it can not. Secondly, a general lower bound for $d_{max}$ is derived. Existence results are proved using the matroid constructions from the previous chapter to construct matroids with a desired $d$. Non-existence results of optimal matroids are proved using Theorem \ref{thmStructure}. As new results, we will extend one class of parameters for which the bound can be achieved as well as give an improved lower bound for $d_{max}$.

First we will need a result on which parameter sets are possible for an $(n,k,d,r,\delta)$-matroid: \newline

\begin{thm}
    (\cite{artikkeli}) Let M = $(E,\rho)$ be an $(n,k,r,\delta)$-matroid. Then we have
    \begin{equation}
    \label{eqParamRaja}
    k \leq n -  \kr (\delta-1).
    \end{equation}
\end{thm}

\begin{proof}
    The inequality can be equivalently expressed as $\kr (\delta-1) \leq \eta(E)$. The left side can be seen to be a lower bound for $\eta(E)$ by Lemma \ref{lemmaChain} (iii) and $m \geq \kr$ in Lemma \ref{lemmaTwoIneqs}.
\end{proof}

We also have $\delta \geq 2$, since $\delta=1$ would allow independent locality sets, which would make local error correction impossible.

Let us now review the results on $d_{max}$ in \cite{artikkeli} and present the two improvements. We will start by stating the results on when the bound can be achieved, after which we will consider a general lower bound for $d_{max}$.

\subsection{Achievability of the generalized Singleton bound}

Exactly an $(n,k)$-uniform matroid achieves the bound when $r=k$. The bound is in this case simplified into $d \leq n-k + 1$. From Theorem \ref{thmParamsFromCyclicFlats} (iii) we see that the bound is achieved if and only if $\mathcal{Z} = \{\emptyset, E\}$, which is satisfied exactly by uniform matroids. Moreover, uniform matroids allow use to choose the required $(r,\delta)$-locality sets. The same uniform matroid that we used for $r=k$ is also valid for $r > k$ when $k$ is the same as before, although choosing such an $r$ has no practical use whatsoever. For the rest of the discussion, we will consider the case $r<k$.

Let $m$ denote the number of atoms of a matroid. In order to achieve the bound, we must have $m \geq \lceil \frac{n}{r+\delta-1} \rceil$, since otherwise we would have an atom $F_i$ with $|F_i| > r+\delta-1$, and condition (ii) in Theorem \ref{thmStructure} would not be satisfied.

Let us now introduce two important constants for an $(n,k,d,r, \delta)$-matroid:

\begin{align}
& a=\kr r -k, \\
& b = \left\lceil \frac{n}{r+\delta-1} \right\rceil (r+\delta-1) - n.
\end{align}

We notice that $\lkr r$, denoted by $k_{max}$, gives the largest possible rank of a union of $\lkr$ atoms, i.e. $F_T$  with $|T|=\lkr$. Also, $\lceil \frac{n}{r+\delta-1} \rceil (r+\delta-1)$, denoted by $n_{max}$, gives the largest possible size of an optimal matroid with $\lceil \frac{n}{r+\delta-1} \rceil$ atoms. We get this kind of a matroid from Theorem \ref{thmCorollIII.1} by choosing $\lceil \frac{n}{r+\delta-1} \rceil$ disjoint sets $F_i$ which are of size $r+\delta-1$, and the parameter $k$ as $k_{max}$. We call such a matroid a \textit{broad matroid}. Any matroid with the parameters $(n,k,d,r,\delta)$ induces a broad matroid, which in turn has the parameters $(n_{max}, k_{max}, d_{opt}, r, \delta)$, where $d_{opt}$ is always such that the broad matroid achieves the bound.

The constants $a$ and $b$ have illustrative interpretations in the context of optimal matroids: $b$ tells us how much smaller the matroid is compared to the size of a corresponding broad matroid. For an interpretation for $a$, remember that an optimal matroid must have $\rho(F_T)=k$ for every $T\subseteq [m]$ with $|T|=\lkr$. Thus $a$ tells how much smaller the rank of such unions of atoms can be compared to that of a broad matroid, in order for the original matroid to be optimal.

For matroids from Theorem \ref{thmCorollIII.1}, there exists an even more useful interpretation for $a$, which is described by Lemma \ref{lemmaOptForCorollIII.1} below. Note that for these matroids, the rank function restricted to cyclic flats, $\rho: \mathcal{Z} \rightarrow \mathbb{Z}$, can be expressed as $\rho(X) = \min\{\rho' (X), k \}$, where $\rho' (F_I) = |F_{I}|-\sum_{i \in I} \eta(F_i)$ for $I \subseteq [m]$. \newline



\begin{lemma}
    \label{lemmaOptForCorollIII.1}
    An $(n,k,d,r,\delta)$-matroid from Theorem \ref{thmCorollIII.1} has $d=n-k+1-(\kr-1)(\delta-1)$ if and only if
    
    \begin{equation}
    \label{eqSuffForOptimality}
    \begin{alignedat}{3}
    &\te{(i)} \quad && |F_T|\geq \kr (r+\delta-1) - a \te{ for each } T \subseteq [m] \te{ with } |T|=\kr,\\
    &\te{(ii)} \quad && \eta(F_i)=\delta-1 \te{ for each } i \in [m].
    \end{alignedat}
    \end{equation}
    
\end{lemma}

\begin{proof}
    Let us first prove that the conditions are sufficient. For each $T \subseteq [m]$ with $|T|=\kr$, we have
    
    \begin{equation*}
    \begin{split}
   & \kr r -k = a \\
   & \geq \kr (r+\delta-1) - |F_T| \\
   & = \kr r - (|F_T|-\sum_{i \in T} \eta(F_i)).
    \end{split}
    \end{equation*}

Thus we have $\rho'(F_T)=|F_T|-\sum_{i \in T} \eta(F_i) \geq k$ and $\max_{F_I \in \mathcal{Z}_{<k}}\{|I|\}\leq \lkr - 1$. This in turn implies $\max\{\sum_{i \in I}\eta(F_i): F_I \in \mathcal{Z}_{<k} \} \leq (\lkr-1)(\delta-1)$ and finally $d \geq n-k+1 - (\lkr-1)(\delta-1)$ via Theorem \ref{thmConstr1Props} (iii).

Now we prove that the conditions in \eqref{eqSuffForOptimality} are necessary. Condition (ii) not being satisfied would contradict with Theorem \ref{thmStructure} (ii) a). If (i) is false and (ii) true, we get $\rho'(F_T)=|F_T|-\sum_{i \in T} \eta(F_i)<k$ for some $T$ with $|T|=\lkr$ in the same manner as above. Together with the result in \cite{artikkeli} that matroids from Theorem \ref{thmCorollIII.1} have $\rho'(F_I)<\rho'(F_J)$ for $I \subsetneq J \subseteq [m]$, we now get $d\leq n-k+1 - \lkr(\delta-1)$, similarly as above.
\end{proof}

Thus for matroids from Theorem \ref{thmCorollIII.1}, the constant $a$ also determines how much smaller unions of $\lkr$ atoms are allowed to be compared to those of a corresponding broad matroid.

 There are two ways in which a set $F_I$ can be smaller than its broad matroid counterpart:
\begin{enumerate}
    \item Its atoms $F_i$ can intersect with each other.
    \item One or more of its atoms $F_i$ may have  $|F_i| < r+\delta-1$, in which case $\rho(F_i)<r$ but $\eta(F_i)=\delta-1$.
\end{enumerate}

It is now easily seen that if $a\geq b$, any matroid from Theorem \ref{thmCorollIII.1} satisfying (ii) in \eqref{eqSuffForOptimality} and having $\lceil \frac{n}{r+\delta-1} \rceil$ atoms is optimal, since $a \geq b = \lceil \frac{n}{r+\delta-1} \rceil(r+\delta-1)- n \geq \lkr (r+\delta-1) - |F_T|$.

However, the situation $b>a$ is more difficult, the general problem being minimizing $\max \{|F_I|: |I| = \lkr\}$, while having $|F_{[m]}|= n$ and $\eta(F_i) = \delta-1$ for each $i \in [m]$.  In a sense, the $b$  ``losses'' in the size of $F_{[m]}$ should be distributed as sparsely and evenly across the matroid as possible, so that as few of them as possible could at most be included in a union of $\lkr$ atoms.

In general it seems favorable to have $|F_i| = r + \delta - 1$ for each $i \in [m]$ and reduce $|F_{[m]}|$ by having the atoms $F_i$ intersect. This is because roughly speaking, picking an atom with reduced size to a union of $\lkr$ atoms automatically decreases the size of the union, but having two atoms intersect requires picking both of them for the size to decrease.

If we allow the atoms $F_i$ to intersect arbitrarily (but each having at least $\eta(F_i)+1$ elements unique to them, required by Theorem \ref{eqCorollIII.1} (iv)), finding the optimal scheme is extremely difficult for most classes of parameters. This is why in \cite{artikkeli}, the scope of the problem is limited by only considering matroids from Graph construction 1 where no 3-cycles are allowed (implying that an element of $E$ belongs to at most two atoms). In this case, the problem is reduced to constructing graphs with minimal $\max_{G'}(\sum_{w \in W_{G'}} \gamma(w))$ where $G'$ is any subgraph induced by $\lkr$ vertices. This viewpoint yields optimal matroids for many classes of parameters $(n,k,r,\delta)$, but for some classes of parameters, it does not even though the Singleton bound could be reached by other matroids.

As an example, the following bound was derived in \cite{artikkeli}, using a version of Graph construction 1, for when $b>a \geq \lceil \frac{k}{r} \rceil -1$ and $\lceil \frac{k}{r} \rceil = 2$. Then we have $d_{max} = n-k+1-\left( \kr -1\right)(\delta-1)$ if

\begin{equation}
\label{eqOldForKr2}
\left\lceil \frac{n}{r+\delta-1} \right\rceil \geq
\begin{cases}
\left\lceil \frac{b}{a} \right\rceil + 1 & \te{ if } 2a \leq r-1, \vspace{1.5 mm}\\
\left\lceil \frac{b}{\left\lfloor \frac{r-1}{2} \right\rfloor} \right\rceil + 1 & \te{ if } 2a > r-1.
\end{cases}
\end{equation}

This bound can be improved however. Previously we viewed our problem as maximizing the minimal union of $\lkr$ atoms while having a certain $n=|E|$. Let us now adopt the reverse viewpoint: While having $\lkr r - a$ as a minimum for the size of a union of $\lkr$ atoms, we minimize $n$.

The original bound was derived using a matroid from Graph construction 1, which means no element can be part of three or more atoms. However, as $\lceil \frac{k}{r} \rceil = 2$, the fulfilment of the condition (i) in \eqref{eqSuffForOptimality} only depends on the sizes of unions of two atoms. Thus we can ``pack'' the atoms in a smaller room by having a set $X$ with $|X|= a$, for which $X \subseteq F_i$ for as many atoms $F_i$ as needed. Exactly this is done in our first result, Theorem \ref{thmResult1}, given below. However, we first need the following lemma: \newline

\begin{lemma}
    \label{suhderaja}
    (\cite{artikkeli}) Let $M$ be an $(n,k,d,r,\delta)$-matroid. Then the following hold:
    \begin{equation*}
    \left\lceil \frac{n}{r+\delta-1} \right\rceil \geq
    \begin{cases}
    \kr \quad \quad \textrm{if $b \leq a$} \vspace{0.2 cm},\\
    \kr +1 \quad \textrm{if $b>a$} .
    \end{cases}
    \end{equation*}
    
    \begin{proof}
        The proof can be found in \cite{artikkeli}. The result follows directly from the inequality \eqref{eqParamRaja}.\newline
    \end{proof}
\end{lemma}

Now we can give the enlarged class of parameters for which the Singleton bound can be achieved: \newline
\begin{thm}
    \label{thmResult1}
    Let $(n,k,r,\delta)$ be integers such that $0<r<k\leq n - \kr(\delta-1)$, $\delta \geq 2$, $b>a \geq \lceil \frac{k}{r} \rceil -1$ and $\lceil \frac{k}{r} \rceil = 2$. If
    \begin{equation}
    \left\lceil \frac{n}{r+\delta-1} \right\rceil \geq \left\lceil \frac{b}{a} \right\rceil + 1,
    \end{equation}
    
    then $d_{max} = n-k+1-\left( \kr -1\right)(\delta-1)$.
    
    \begin{proof}
        We prove our result by giving an explicit construction for matroids which achieve the bound when the conditions are satisfied.
        
        \emph{A matroid construction.} Let $n'$, $r'$, $\delta'$ and  $k$ be integers such that $0<r'<k\leq n' - \left\lceil \frac{k}{r'} \right\rceil(\delta'-1)$, $\delta' \geq 2$, $b'>a'$, $m \geq \left\lceil \frac{b'}{a'} \right\rceil + 1$, where we define
        
        \begin{align*}
        & b' = \left\lceil \frac{n'}{r'+\delta'-1} \right\rceil (r'+\delta'-1) - n' , \\
        & a'=\left\lceil \frac{k}{r'} \right\rceil r' -k ,\\
        & m = \left\lceil \frac{n'}{r'+\delta'-1} \right\rceil.
        \end{align*}
        
        Let $F_1, ..., F_m = \{F_i\}_{i \in [m]}$ be a collection of finite sets with $E = \bigcup_{i \in m}F_i$ and $X \subseteq E$ a set such that
        
        \begin{equation}
        \label{eqSetConditions}
        \begin{alignedat}{3}
        &\te{(i)} \quad && F_i \cap F_j \subseteq X \quad \te{ for }i,j \in [m] \te{ with } i \neq j, \\
        &\te{(ii)} \quad && |X| = a', \\
        &\te{(iii)} \quad && |F_i| = r' + \delta' - 1 \quad \te{ for } i \in [m], \\
        &\te{(iv)} \quad && |F_i \cap X| = a' \te{ for } 1 \leq i \leq \left\lceil \frac{b'}{a'} \right\rceil, \\
        &\te{(v)} \quad && |F_i \cap X| = b'- \left(\left\lceil \frac{b'}{a'} \right\rceil  - 1 \right)a' \te{ for } i = \left\lceil \frac{b'}{a'} \right\rceil + 1,\\
        &\te{(vi)} \quad && |F_i \cap X| = 0 \te{ for } i > \left\lceil \frac{b'}{a'} \right\rceil + 1. \\
        \end{alignedat}
        \end{equation}
        
        Let $\rho$ be a function $\rho:   \{F_i\}_{i \in [m]} \rightarrow \mathbb{Z}$ such that $\rho(F_i)=r'$ for each $i \in [m]$.
        
        Now we prove that this gives a matroid obtainable from Theorem \ref{thmCorollIII.1}. The set $E$ is clearly finite and the conditions (i) and (ii) of \eqref{eqCorollIII.1} are trivially satisfied. Let us next calculate the size of $F_{[m]}=E$ by adding each of the sets $F_i$ to the union one at a time, in order according to their index $i$. Let us write $s = r'+\delta'-1$.  The first set $F_1$ adds $|F_1| =s$ elements. As $X \subseteq F_1$ and the sets in $\{F_i\}_{i \in [m] \setminus \{1\}}$ only intersect with $X$, each subsequent set $F_i$ adds $|F_i|- |F_i \cap X|$ elements. Thus we get the following: The sets $F_i$ for $2 \leq i \leq \lceil \frac{b'}{a'} \rceil$ each add $s-a'$ new elements. The set $F_{\lceil \frac{b'}{a'} \rceil+1}$ adds $s - (b'- (\lceil \frac{b'}{a'} \rceil  - 1 )a')$ new elements. The rest of the sets add $s$ elements each. Simple cancellation of terms then gives us
        
        \begin{equation*}
        \begin{alignedat}{2}
        |F_{[m]}|& = s + \left( \left\lceil \frac{b'}{a'} \right\rceil -1\right)(s-a') + s - \left(b'- \left(\left\lceil \frac{b'}{a'} \right\rceil  - 1 \right)a'\right) + \left(\left\lceil \frac{n'}{s} \right\rceil - \left(\left\lceil \frac{b'}{a'} \right\rceil + 1\right)\right) s \\
        & = \left\lceil \frac{n'}{s} \right\rceil s - b' \\
        & = n'.
        \end{alignedat}
        \end{equation*}

        Using $s-a'$ as a lower bound for the number of elements added by the sets $F_i$ with $3 \leq i  \leq \lceil \frac{n'}{s} \rceil$ and recalling that $r' > a$ and $\lceil \frac{k}{r'} \rceil = 2$, we get the following:

        \begin{equation*}
        \begin{alignedat}{2}
        &|F_{[m]}|-\sum_{i \in [m]} \eta(F_i) \\
        \geq & \ s+s-a' + \left(\left\lceil \frac{n'}{s} \right\rceil -2 \right)(s-a') - \left\lceil \frac{n'}{s} \right\rceil (\delta'-1) \\
        = &  \ 2r' - a' + \left(\left\lceil \frac{n'}{s} \right\rceil -2 \right) (r'-a')\\
        \geq & \left\lceil \frac{k}{r'} \right\rceil r' -a' \\
        = & \ k.
        \end{alignedat}
        \end{equation*}
        
        This shows that the construction satisfies \eqref{eqCorollIII.1} (iii). When I $\subseteq [m], j \in [m] \setminus I$, we have $|F_I \cap F_j| \leq |X \cap F_j|\leq a' < r' = \rho(F_j)$, so \eqref{eqCorollIII.1} (iv) is also satisfied. Thus the construction gives a matroid obtainable from Theorem \ref{thmCorollIII.1}.
        
        According to Theorem \ref{thmConstr1Props} we now have
        
        \begin{equation*}
        \begin{alignedat}{3}
        & \te{(i)} \quad && n = |E| = n', \\
        & \te{(ii)} \quad && k = \rho(E), \\
        & \te{(iii)} \quad && d = n-k+1-(\delta-1) = n-k+1-\left( \kr -1\right)(\delta-1), \\
        & \te{(iv)} \quad && \delta = \delta', \\
        & \te{(v)} \quad && r = r'. \\
        \end{alignedat}
        \end{equation*}

        Note that (i),(ii), (iv) and (v) imply that $a=a'$ and $b=b'$ so we can stop using the primed letters. For (iii), note that for every $F_i, F_j \in \{F_t\}_{t \in [m]}$ we have
        
        \begin{equation*}
        \begin{alignedat}{2}
        &\rho(F_{\{i,j\}} ) =  |F_{\{i,j \}}|-\sum_{i \in \{i,j\}} \eta(F_i) \\
        = \ & |F_{\{i,j\}}| - 2(\delta-1) \\
        \geq \ & |F_i|+|F_j| - |X| - 2(\delta-1) \\
        = \ & \rho(F_i) + \rho(F_j) - a \\
        = \ & r+r-(2r-k) \\
        = \ & k,
        \end{alignedat}
        \end{equation*}
        
        which implies that for any $i,j$, $F_{\{i,j\}}  \notin \mathcal{Z}_{<k}$. Due to \eqref{rankEqs} (ii), for any unions $F_I$ with $|I| \geq 2$ we have $F_I \notin \mathcal{Z}_{<k}$. Thus $\mathcal{Z}_{<k} = \{F_i\}_{i \in [m]}$ as $r<k$.
        
        Now we have shown that the construction gives a matroid that achieves the generalized Singleton bound for every desired parameter set $(n,k,d,r,\delta)$.

    \end{proof}
\end{thm}

Even this scheme does not give the full class of matroids which are optimal and have $\lkr = 2$. To see how the atoms could be organized even more efficiently, notice that we may have $a<r-1$, in which case the atoms of our construction have more non-shared elements than would be necessary. Using these unnecessarily non-shared elements, the atoms $F_i$ could be packed even more efficiently.

For example, if we denote by $F_i'$ the set of elements the atom $F_i$ is allowed to share ($|F_i'|=\rho(F_i)-1$), the following setup of three atoms $F_i'$ would be possible, when $a=2$ and $|F_i'|=4$:

\begin{equation}
\label{bestSchemeExample}
\begin{alignedat}{2}
& F_1' = \{1,2,3,4 \}, \\
& F_2' = \{1,2,5,6\}, \\
& F_3' = \{3,4,5,6\}.
\end{alignedat}
\end{equation}

Here we have $n=6+3(\eta(F_i)+1)$. For a reference, the construction used in Theorem \ref{thmResult1} in a similar case would be

\begin{align*}
& F_1' = \{1,2,3,4\}, \\
& F_2' = \{1,2,5,6\}, \\
& F_3' = \{1,2,7,8\},
\end{align*}

where $n=8+3(\eta(F_i)+1)$. Here we have $2a \leq r-1$, so the construction for \eqref{eqOldForKr2} would also have $n=8+3(\eta(F_i)+1)$.

Graph construction 1 is also used to derive classes of optimal matroids for $\lkr \geq 3$. Improvements similar to Theorem \ref{thmResult1} may be possible here as well, but they would be considerably more complicated.

\subsection{A general lower bound for $d_{max}$}

Besides finding the classes of parameters $(n,k,r,\delta)$ for which the bound can be achieved, we are also interested in finding $d_{max}$ for those classes for which the bound can not be achieved. The following partial result towards this goal was presented in \cite{artikkeli}, for $b>a$:

\begin{equation}
\label{oldbound}
d_{max} \geq \begin{cases}
n-k+1-\kr(\delta-1) &\textrm{if $b \leq r-1,$} \\
n-k+1-\kr(\delta-1)+(b-r) &\textrm{if $b \geq r$.}
\end{cases}
\end{equation}

If we do not require optimality, we can ignore some of the requirements in \eqref{eqSuffForOptimality}. If we let go of (i) and allow that we may reach full rank only with a union of $\lkr+1$ atoms, we get the bound for $b \leq r-1$ in \eqref{oldbound}.

If we let go of (ii) we can use $m = \left\lceil \frac{n}{r+\delta-1} \right\rceil - 1$, which lets us have atoms of full rank $r$ that are pairwise disjoint. In order that the union of the atoms is large enough we must however have at least one atom with $\eta(F_i)> \delta-1$. This was done in \cite{artikkeli} to obtain the bound for $b \geq r$ in \eqref{oldbound}. The corresponding matroid construction had one atom that contained all the extra nullity required. However, a better bound can be derived by spreading the extra nullity evenly among the atoms, since then all the extra nullity can not be included in a coatom, i.e. a union of $\kr-1$ atoms. This is done in the following theorem:  \newline

\begin{thm}
    \label{raja}
    Let $(n,k,r,\delta)$ be integers such that $0<r<k\leq n - \kr(\delta-1)$, $\delta \geq 2$ and $b>a$. Also let $m=\left\lceil \frac{n}{r+\delta-1} \right\rceil - 1$ and $v = r+\delta-1-b-\left\lfloor \frac{r+\delta-1-b}{m} \right\rfloor m$.
    
    Then, if $\delta-1 \leq \left( \left\lceil\frac{k}{r} \right\rceil -1\right) \left\lfloor \frac{r+\delta-1-b}{m} \right\rfloor
    + \min\{v, \kr -1\}$, we have
    \begin{equation}
    \label{remainingbound}
    d_{max} \geq n-k+1-\kr(\delta-1).
    \end{equation}
    
    Otherwise, if $\delta-1 > \left( \left\lceil\frac{k}{r} \right\rceil -1\right) \left\lfloor \frac{r+\delta-1-b}{m} \right\rfloor
    + \min\{v, \kr -1\}$, then
    
    \begin{equation}
    \label{newbound}
    d_{max} \geq n-k+1- \left( \kr -1\right)\left(\left\lfloor \frac{r+\delta-1-b}{m} \right\rfloor+\delta-1\right)  - \min\left\{v, \kr -1\right\}.
    \end{equation}
    
    This bound is always at least as good as the bound in \eqref{oldbound}. Moreover, denoting the bound for $b \geq r$ in \eqref{oldbound} by $d_{old} = n-k+1-\kr(\delta-1)+(b-r)$ and similarly the bound in \eqref{newbound} by
    
    \begin{equation*}
    d_{new} = n-k+1- \left( \kr -1\right)\left(\left\lfloor \frac{r+\delta-1-b}{m} \right\rfloor+\delta-1\right)  - \min\left\{v, \kr -1\right\},
    \end{equation*}
    
    it follows that
    
    \begin{equation}
    \label{parannus}
    d_{new} - d_{old} \geq \left\lfloor \frac{r+\delta-1-b}{m} \right\rfloor \left(m-\left\lceil\frac{k}{r} \right\rceil +1\right) \geq 0.
    \end{equation}
    
\end{thm}

\begin{proof}
    Let $n' \in \mathbb{Z}$ be such that it satisfies the conditions for n in Theorem \ref{raja}. Then, analogously to the result in Lemma \ref{suhderaja}, $n'$ also satisfies $\left\lceil \frac{n'}{r+\delta-1} \right\rceil  \geq  \kr +1$.
    
    \emph{Graph construction}. Let $G(\alpha, \beta, \gamma; k, r, \delta)$ be intended as an instance of Graph construction 1 with
    
    \begin{align}
    \begin{split}
    & \textrm{(a)} \quad m = \left\lceil \frac{n'}{r+\delta-1} \right\rceil - 1,\\
    & \textrm{(b)} \quad  W=\emptyset, \\
    & \textrm{(c)} \quad \alpha (i)= 0 \textrm{ for } i \in [m], \\
    & \textrm{(d)} \quad \beta (i) = \begin{cases} \left\lceil \frac{r+\delta-1-b'}{m} \right\rceil \textrm{ for } 1 \leq i \leq v' \vspace{1.5 mm}, \\
    \left\lfloor \frac{r+\delta-1-b'}{m} \right\rfloor \textrm{ for } v' <  i \leq m,
    \end{cases}
    \end{split}
    \end{align}

    where $b' = \left\lceil \frac{n'}{r+\delta-1} \right\rceil (r+\delta-1) - n'$ and $v' = r+\delta-1-b'-\left\lfloor \frac{r+\delta-1-b'}{m} \right\rfloor m$.  Now we have to show that the conditions in equation \eqref{c1} apply in order to prove that this is indeed an instance of Graph construction 1.  Conditions (i), (ii) and (iv) are trivially satisfied. As $b' < r+\delta-1$, (iii) is also true. For our construction G, requirement (v) is simplified to the form $k \leq rm$, which is true as
    
    \begin{equation*}
    k=\kr r - a \leq \kr r \leq \left(\left\lceil \frac{n'}{r+\delta-1} \right\rceil - 1\right) r = mr.
    \end{equation*}
    
    As $W= \emptyset$, $\alpha (i)= 0 \textrm{ for } i \in [m]$ and $r>0$, (vi) is
    also true. Thus G is an instance of Graph construction 1.

    We have $\sum_{i \in [m]}\beta(i) = r+\delta-1-b'$ because of the following:  If  $\frac{r+\delta-1-b'}{m}  \in \mathbb{Z}$:
    
    \begin{equation*}
    \sum_{i \in [m]}\beta(i) = m \cdot \frac{r+\delta-1-b'}{m} = r+\delta-1-b'.
    \end{equation*}
    
    If $\frac{r+\delta-1-b'}{m}  \notin \mathbb{Z}$:
    
    \begin{equation*}
    \begin{split}
    \sum_{i \in [m]}\beta(i) &= \left\lceil \frac{r+\delta-1-b'}{m} \right\rceil  v' + \left\lfloor \frac{r+\delta-1-b'}{m} \right\rfloor (m-v') \\
    &= v' + \left\lfloor \frac{r+\delta-1-b'}{m} \right\rfloor m \\
    &= r+\delta-1-b'.
    \end{split}
    \end{equation*}

    From Theorem \ref{c1thm} we get thus that
    \begin{equation*}
    \begin{split}
    n&=(r+\delta-1)m + \sum_{i \in [m]}\beta (i) \\
    &= (r+\delta-1)(m+1)-b' \\
    &= (r+\delta-1)\left\lceil \frac{n'}{r+\delta-1} \right\rceil-\left( \left\lceil \frac{n'}{r+\delta-1} \right\rceil (r+\delta-1) - n'\right) \\
    &= n'.
    \end{split}
    \end{equation*}
    
    This shows that we can use this construction for any desired parameter set $(n,k,r,\delta)$ satisfying the requirements in Theorem \ref{raja}. As $n=n'$ also $a=a'$, $b=b'$ and $v=v'$, so we can from now on only use the non-primed letters.
    
    Using Theorem \ref{c1thm} (ii), we show next that this construction gives a $d$ that is the desired lower bound for $d_{max}$. We have $\max_{I \in V_{<k}}(|I|) = \lkr -1 $, as  for every $I \in [m], \sum_{i \in I} \alpha(i) = \sum_{w \in I \times I} \gamma(w) = 0$. Clearly $(\delta-1)|I|+\sum_{i \in I} \beta(i)$ is maximized with a maximal $|I|$, so we get from Theorem \ref{c1thm} (ii) that
    
    \begin{equation*}
    d = n-k+1-
    \left( \kr -1 \right) (\delta-1) -
    \left( \kr -1\right)\left\lfloor \frac{r+\delta-1-b}{m} \right\rfloor  - \min\left\{v, \kr -1\right\}.
    \end{equation*}
    
    We now get the bound in \eqref{newbound} by rearranging the right side of this equation.
    
    Next we prove the result in \eqref{parannus}. By simple cancellation of identical terms and some rearranging we get the following inequality that is equivalent to the first inequality in \eqref{parannus}:
    
    \begin{equation*}
    m \left\lfloor \frac{r+\delta-1-b}{m} \right\rfloor + \min\left\{v, \kr -1\right\} \leq r+\delta-1-b.
    \end{equation*}
    
    This inequality can be seen to be true by using $v$ as an upper bound for $\min\{v, \lkr -1\}$ and substituting $v$ by its definition. Thus the first inequality \eqref{parannus} is true. The second inequality in \ref{parannus} is true because $m = \lceil \frac{n}{r+\delta-1} \rceil - 1 \geq \lkr$  by Lemma \ref{suhderaja}.
    
    The combined bound of \eqref{remainingbound} and \eqref{newbound} is always at least as good as \eqref{oldbound} because $d_{new} \geq d_{old}$ and we always use the stricter of the two bounds \eqref{remainingbound} and \eqref{newbound} in our result.\newline

\end{proof}

Could this bound be improved even further? The following theorem gives a partial answer by stating that for matroids from Theorem \ref{thmCorollIII.1}, the bound is strict for parameter sets $(n,k,r,\delta)$ for which there exists no optimal matroid from Theorem \ref{thmCorollIII.1}. \newline 

\begin{thm}
    Let $(n,k,r,\delta)$ be integers such that there exists no $(n,k,d',r,\delta)$-matroid from Theorem \ref{thmCorollIII.1} with $d'=n-k+1-(\lkr-1)(\delta-1)$. Let $M$ be a $(n,k,d,r,\delta)$-matroid from Theorem \ref{thmCorollIII.1} and let us denote the bound in Theorem \ref{raja} by $d_{b}=d_{b}(n,k,r,\delta)$. Then $d \leq d_b$.

\end{thm}

\begin{proof}
    Let $M=M(F_1,...,F_m;k;\rho)$ be a matroid from Theorem \ref{thmCorollIII.1} for which there exists no optimal matroid from Theorem \ref{thmCorollIII.1} with the same parameters $(n,k,r,\delta)$.
    
    Assume that $\max\{|I| : F_I \in \mathcal{Z}_{<k}\} \geq \lkr$. Using Theorem \ref{thmConstr1Props} (iii), we then obtain $ d \leq n-k+1-\kr(\delta-1)$, as $\eta(F_i) \geq \delta-1$ for every $i \in [m]$. 
    
    Thus the theorem holds in this case and we are only left with the case $\max\{|I| : F_I \in \mathcal{Z}_{<k}\} = \lkr - 1$. Having $\max\{|I| : F_I \in \mathcal{Z}_{<k}\}<\lkr-1$ is impossible, since it would imply $\rho(F_J)=k$ for every $J$ with $|J|=\lkr-1$. This in turn is impossible since $\rho(F_i) \leq r$ and rank behaves subadditively, according to  \eqref{rankEqs}(iii).
    
    There must be an atom $F_i$ with $\eta(F_i) > \delta-1$, since otherwise the matroid would be optimal. Next we show that our current assumptions imply $m < \lceil \frac{n}{r+\delta-1} \rceil$. We do this by showing that $m \geq \lceil \frac{n}{r+\delta-1} \rceil$ would allow the existence of optimal matroids, which is a contradiction. The optimal matroids are constructed by repeatedly applying Algorithm 1, which takes a Theorem \ref{thmCorollIII.1} -matroid $M_i$ satisfying
    
    \begin{enumerate}
    	\item $\max\{|I| : F_I \in \mathcal{Z}_{<k}\} = \lkr - 1$,
    	\item $m \geq \lceil \frac{n}{r+\delta-1} \rceil$,
    	\item $\exists F_u \in A_{\mathcal{Z}(M_i)}: \ \eta(F_u) > \delta-1$ ,
    \end{enumerate} 
    
    and returns another Theorem 9 -matroid $M_{i+1}$ still satisfying 1. and 2. but having the nullity of the atom $F_u$ with $\eta(F_u) > \delta-1$ reduced by one.
    
    \begin{algorithm}
        \caption{From $M_i$ to $M_{i+1}$}
    \begin{algorithmic}[1]
        \State $F_u$ is an atom with $\eta(F_u)>\delta-1$.
        \State $x$ is an element $x \in F_u$.
        \State $F_u \gets F_u \setminus \{x\}$
        \If {$\forall i \in [m]: x \notin F_i$}
            \If {$\exists F_j: \rho(F_j)<r$}
                \State $F_j \gets F_j \cup \{x\}$
                \State $\rho(F_j) \gets \rho(F_j)+1$
            \Else
                \State $F_k, F_l$ are distinct atoms with $|F_k \cap F_l| \geq 1$.
                \State $y$ is an element $y \in F_k \cap F_l$.
                \State $F_k \gets F_k \setminus \{y\}$
                \State $F_k \gets F_k \cup \{x\}$
            \EndIf
            
        \EndIf
    \end{algorithmic}
    \end{algorithm}
    
    The definition of the algorithm is otherwise clearly sound, but we need to prove that the atoms $F_k$ and $F_l$ required on line 9 exist. Assume on the contrary that they do not. Then,
    \begin{equation*}
    \left|\bigcup_{i \in [m]}F_i\right| > \left\lceil \frac{n}{r+\delta-1} \right\rceil (r+\delta-1) \geq n,
    \end{equation*}
    
    which is contradiction. Thus the desired $F_k$ and $F_l$ exist.
    
    Next, we need to prove that $M_{i+1}$ is a matroid from Theorem \ref{thmCorollIII.1}. We do this by proving that the conditions (i)-(iv) of \eqref{eqCorollIII.1} are satisfied. 
    
    Condition (i) stays satisfies after executing lines 1-3 as originally $|F_u| = \rho(F_u) + \eta(F_u) > \rho(F_u)+1$. Lines 5-7 preserve the condition as both the rank and the size of the atom are increased by one. Lines 9-12 preserve the size and rank of $F_k$ and $F_l$. Condition (ii) is satisfied as $x$ is re-added and $y$ always remains in $F_l$. Condition (iii) is satisfied for $M_{i+1}$ as it is satisfied for $M_i$, and the nullity of an atom is never increased and $|F_{[m]}|=|E|$ stays constant.
    
    
    To prove (iv), first note that the condition is equivalent to every atom $F_i$ having at least $\eta(F_i)+1$ elements that are not contained in any other atom, i.e. are non-shared. 
    On lines 1-3, $F_u$ may lose one non-shared element, but its nullity also decreases by one. Also, $x$ can not be unique to any other atom besides $F_u$. On lines 5-6, $x$ does not belong to any atom and we do not increase the nullity or decrease the amount of unique elements of $F_j$. Similarly on lines 10-11, $y$ is not unique to any atom and we do not alter the nullities or decrease the amounts of unique elements of $F_k$ or $F_l$. Thus $M_{i+1}$ satisfies (iv) and we have proved that $M_{i+1}$ is a matroid from Theorem \ref{thmCorollIII.1}.
    
    The new matroid $M_{i+1}$ also satisfies $\max\{|I| : F_I \in \mathcal{Z}_{<k}\} = \lkr - 1$. This is because the nullity of an atom is never increased, and $F_u$ is the only atom whose size is reduced, but this is compensated by correspondingly reducing its nullity. Thus $\rho'_{i+1}(F_I) \geq \rho'_i(F_I)$ for every $I \subseteq [m]$, where we denote $\rho'(F_I) = |F_{I}|-\sum_{i \in I} \eta(F_i)$, and the subscript distinguishes between the matroids $M_i$ and $M_{i+1}$. 
    
    
    
    Thus we can repeatedly use Algorithm 1 to decrease the nullity of some atom $F_u$ with $\eta(F_u)>\delta-1$ until every atom has $\eta(F_i)=\delta-1$. From Theorem \ref{thmConstr1Props} (iii) we then see that we have obtained an optimal matroid. However, this is a contradiction and therefore $m < \left\lceil \frac{n}{r+\delta-1} \right\rceil$.

    Let us denote $s=\sum_{i \in [m]} \eta(F_i)$. Let us distribute this nullity evenly among the atoms $F_i$, i.e. set
    
    \begin{equation}
    \eta (F_i) = \begin{cases} \left\lceil \frac{s}{m} \right\rceil \textrm{ for } 1 \leq i \leq s- \left\lfloor \frac{s}{m} \right\rfloor m \vspace{1.5 mm}, \\
    \left\lfloor \frac{s}{m} \right\rfloor \textrm{ for } s- \left\lfloor \frac{s}{m} \right\rfloor m <  i \leq m.
    \end{cases}
    \end{equation}
    
    For minimizing $\max\left\{\sum_{i \in I}\eta(F_i): |I|= \kr-1 \right\}$, this setup is optimal and yields the bound
    
    \begin{equation}
    \label{eqBoundTasajako}
    \begin{split}
    & \max\left\{\sum_{i \in I}\eta(F_i): |I|= \kr-1 \right\} \\
    & \geq \left(\kr-1\right)\left\lfloor \frac{s}{m} \right\rfloor + \min\left\{\kr-1, s- \left\lfloor \frac{s}{m} \right\rfloor m  \right\}.
    \end{split}
    \end{equation}
    
    A proof can be given by contradiction: Assume that there exists a set of $\lkr-1$ atoms whose sum of nullities is maximal, but lower than the bound in  \eqref{eqBoundTasajako}. This imposes such an upper bound on the nullities of the other atoms that $\sum_{i \in [m]}\eta(F_i)$ must be lower than $s$, which is of course a contradiction.
    
    The bound in \eqref{eqBoundTasajako} is increasing as a function of $s$. Let us show this by considering how the value of the bound changes when the value of $s$ is increased by one. If $\lfloor \frac{s}{m} \rfloor$ remains unchanged, the value of the bound clearly does not decrease. If the value of $\lfloor \frac{s}{m} \rfloor$ is increased, the first term is increased by $(\lkr-1)$, whereas the value of the second term is altered by at most $(\lkr-1)$, since $0 \leq \min\{\lkr-1, s- \lfloor \frac{s}{m} \rfloor m  \}\leq \lkr-1$. Thus an increment of $s$ by one never decreases the value of the bound and it is increasing as a function of $s$.
    
    We have 
    
    \begin{equation}
    \sum_{i \in [m]} |F_i|  =   \sum_{i \in [m]} \rho(F_i) + \sum_{i \in [m]} \eta(F_i) \geq |E|,
    \end{equation}
    
   so $s \geq n-rm$. As the bound in \eqref{eqBoundTasajako} is increasing as a function of $s$, we obtain the bound
   
   \begin{equation}
   \label{eqBoundTasajako2}
   \begin{split}
   & \max\left\{\sum_{i \in I}\eta(F_i): |I|= \kr-1 \right\} \\
   & \geq \left(\kr-1\right)\left\lfloor \frac{n-rm}{m} \right\rfloor + \min\left\{\kr-1, n-rm- \left\lfloor \frac{n-rm}{m} \right\rfloor m  \right\}.
   \end{split}
   \end{equation}

    By a similar consideration as above, we note that this bound is decreasing as a function of $m$. Thus we can obtain an a new bound by substituting $m = \left\lceil \frac{n}{r+\delta-1} \right\rceil-1$. This is also the definition of $m$ in Theorem \ref{raja}. By additionally substituting $v$ and $b$ by their definitions in \eqref{newbound}, we can see that the bounds \eqref{newbound} and \eqref{eqBoundTasajako2} are equal.
    
    We have thus proved that the value of $d$ for non-optimal matroids is always bounded from above by either the bound \eqref{remainingbound} or the bound \eqref{newbound}. This proves the theorem.

\end{proof}

\newpage


\section{Conclusions}

In this paper, we reviewed several results first established in \cite{artikkeli}: We first demonstrated how almost affine codes induce a matroid in such a way that the key parameters $(n,k,d,r,\delta)$ of a locally repairable code appear as matroid invariants. We then discussed how this enables us to use matroid theory to study properties of almost affine locally repairable codes. We extended the generalized Singleton bound to matroids after which we derived the structure theorem stating a list of requirements for a matroid to achieve the bound. We reviewed results on $d_{max}$ for different classes of parameters $(n,k,r,\delta)$ and gave matroid constructions used to prove these results. Lastly, we presented two improvements to previous results: We extended the class of parameters for which the bound can be achieved when $\lkr = 2$. We also presented an improved general lower bound for $d_{max}$.

There  still remains significant gaps in our knowledge on the complete function $d_{max}(n,k,r,\delta)$. We particularly lack results on when the Singleton bound can not be achieved, except for the class of parameters $a<\kr-1$ for which the reachability of the bound is completely solved in \cite{artikkeli}. Further such non-existence results could be derived using Theorem \ref{thmStructure}.

The classes of parameters for which the bound can be achieved could be extended as well, using matroids from Theorem \ref{thmCorollIII.1} without restricting ourselves to Graph Construction 1, as demonstrated in the example in the previous chapter. However, designing such matroid constructions probably becomes increasingly complicated as we progress in finding improvements. Finding a general optimal scheme would probably be exceedingly difficult. Some related problems are studied by a branch of mathematics called extremal set theory, see for instance \cite{erdos}.

Whether there exists a better general lower bound than the one derived is unknown. We proved that to have a chance at finding such a bound, one needs to use a class of matroids more general than the class given by Theorem \ref{thmCorollIII.1}. However, it seems plausible that our lower bound actually gives the best possible $d$ for the whole class of parameters $(n,k,r,\delta)$ for which the generalized Singleton bound can not be achieved.

Many interesting areas of research remain open in the wider context of almost affine codes and locally repairable codes. For instance, the class of matroids induced by almost affine codes could possibly be expanded from that given in Theorem \ref{thmCorollIII.1}. Also, the exact size of the finite field required to represent matroids from Theorem \ref{thmCorollIII.1} is unknown.

\newpage

\appendix

\end{document}